\documentclass[11pt]{article}
\usepackage{graphicx}
\usepackage{hyperref}
\usepackage{geometry}
\usepackage{natbib}
\usepackage{color}
\usepackage{amsmath, amssymb, graphicx, booktabs}
\usepackage{url}
\usepackage{amsthm}
\newtheorem{theorem}{Theorem}
\newtheorem{remark}{Remark}

\title{Decomposing Global AUC into Cluster-Level Contributions for Localized Model Diagnostics}
\author{Agus Sudjianto$^{1,2}$ and Alice J. Liu$^3$\\
$^1$H2O.ai, \texttt{agus.sudjianto@h2o.ai} \\
$^2$Center for Trustworthy AI Through Model Risk Management, \\
University of North Carolina Charlotte\\
$^3$2nd Order Solutions, \texttt{alice.liu@2os.com}}
\date{}

\begin{document}

\maketitle

\begin{abstract}
The Area Under the ROC Curve (AUC) is a widely used performance metric for binary classifiers. However, as a global ranking statistic, the AUC aggregates model behavior over the entire dataset, masking localized weaknesses in specific subpopulations. In high-stakes applications such as credit approval and fraud detection, these weaknesses can lead to financial risk or operational failures. In this paper, we introduce a formal decomposition of global AUC into intra- and inter-cluster components. This allows practitioners to evaluate classifier performance within and across clusters of data, enabling granular diagnostics and subgroup analysis. We also compare the AUC with additive performance metrics such as the Brier score and log loss, which support decomposability and direct attribution. Our framework enhances model development and validation practice by providing additional insights to detect model weakness for model risk management.
\end{abstract}

\section{Introduction and Motivation}

The Area Under the Receiver Operating Characteristic Curve (AUC-ROC) is a widely used metric for evaluating binary classification models \cite{fawcett2006introduction}. It reflects the probability that a randomly selected positive instance will be ranked above a randomly selected negative one, offering a threshold-independent view of model discrimination performance. However, the AUC is a global statistic: it aggregates information from all positive-negative pairs across the entire dataset. As a result, it can obscure meaningful differences in model performance across specific subgroups or data segments.

In many real-world applications—particularly in high-stakes, regulated domains—such global performance measures are insufficient. Consider the case of \textbf{credit approval} models. A classifier may achieve high overall AUC, but fail to properly rank applicants from a thin file group (i.e., limited credit history) or income segment. This localized deficiency could result in systematic denial of credit to otherwise creditworthy individuals, with legal and reputational consequences for unfair lending practices. Evaluating subgroup-specific performance is crucial for fair and accountable machine learning systems \cite{barocas2019fairness}. When subgroup performance diverges, trade-offs between calibration, fairness, and AUC can arise \cite{kleinberg2016inherent}.

Similarly, in \textbf{fraud detection}, a global AUC of 0.95 may suggest excellent performance. Yet, within certain subpopulations—such as elderly customers using mobile banking or international transactions from high-risk countries—the model may perform poorly, leading to undetected fraud or a high rate of false positives that erode customer trust. These cluster-level failure modes are invisible under a global AUC assessment.

To address these limitations, we propose a framework for \textbf{cluster-level analysis of the AUC}. We treat the dataset as partitioned into meaningful clusters—either defined by domain knowledge (e.g., demographics, channels, regions) or learned via unsupervised methods—and decompose the global AUC into \textbf{intra-} and \textbf{inter-cluster} contributions. This decomposition allows us to evaluate not only how well the model separates positives from negatives within each cluster, but also how well it generalizes across clusters.

\section{Related Work and Contribution}

The AUC is a widely adopted metric in binary classification, with extensions for multiclass use, and is celebrated for its threshold-invariant measure of ranking performance. While the foundational properties of the AUC are well-established \cite{fawcett2006introduction, hanley1982meaning}, its decomposition into intra- and inter-cluster contributions has received limited attention in the supervised learning literature.

\paragraph{AUC and Pairwise Ranking}

The probabilistic interpretation of the AUC as the likelihood that a randomly selected positive instance ranks higher than a negative instance connects it to the Mann–Whitney U-statistic. This equivalence has enabled theoretical work on AUC consistency and optimization \cite{clemenccon2008ranking}. However, most of this work focuses on global properties or algorithmic aspects of AUC optimization, not diagnostics across subgroups.

\paragraph{AUC in Clustering Evaluation}

A parallel line of work uses AUC-like metrics to evaluate clustering performance. Jaskowiak et al.~\cite{jaskowiak2020area} introduced the AUC for Clustering (AUCC), which measures the extent of similarity among items that are grouped together. These metrics, while conceptually related to the AUC, are unsupervised and not concerned with model ranking across known class labels.

\paragraph{Subgroup and Localized AUC}

Recent work in model fairness and reliability has prompted research into \textbf{localized and subgroup-specific AUC analysis}. Carrington et al.~\cite{carrington2022deep} introduced \emph{Deep ROC Analysis}, which partitions the ROC curve and computes AUCs across different subgroups to identify performance disparities. These approaches align with our motivation: to understand AUC behavior across subpopulations. In contrast to these segment-level diagnostics, our framework offers a principled mathematical decomposition into intra- and inter-cluster AUC components.

Additionally, research on the \textbf{partial AUC} (pAUC) focuses on model evaluation over constrained operating regions (e.g., high specificity). While the pAUC is useful in applications like medical diagnostics or fraud detection, it does not provide insights into structured subgroup or cluster-level behavior \cite{dodd2003partial, mcclish1989analyzing}.

\paragraph{Our Contribution}

To our knowledge, this paper presents the first formal decomposition of global AUC into weighted contributions from intra- and inter-cluster AUCs. This enables:
\begin{itemize}
    \item Identification of subgroups where the model fails to discriminate effectively (via low intra-cluster AUC),
    \item Detection of systemic cross-cluster misranking (via low inter-cluster AUC), and
    \item Fine-grained attribution of performance contributions beyond global AUC.
\end{itemize}

Our work generalizes localized AUC analysis and extends it with a mathematically grounded decomposition, providing a robust framework for performance auditing, fairness diagnostics, and subgroup reliability analysis to uncover performance weaknesses hidden in global metrics.

\section{Decomposing Global AUC into Cluster-Level Contributions}

While the AUC provides a useful summary of model ranking performance, it is not an additive metric and can obscure substantial differences in model behavior across subgroups. Even when each individual cluster performs well internally, the global AUC can be degraded by inconsistent rankings between clusters. This motivates a formal decomposition of AUC into intra- and inter-cluster components.

\subsection{AUC as a Pairwise Ranking Statistic}
The AUC can be interpreted as the probability that a randomly chosen positive instance is ranked higher than a randomly chosen negative instance \cite{hanley1982meaning}.
Let $\mathcal{P}$ and $\mathcal{N}$ denote the sets of positive and negative instances, respectively. For a probabilistic classifier $\hat{f}(x)$, the AUC can be written as:

\begin{equation}
\text{AUC} = \mathbb{P}\left(\hat{f}(x^+) > \hat{f}(x^-)\right), \quad x^+ \in \mathcal{P},\ x^- \in \mathcal{N}
\end{equation}

This definition shows that the AUC is a \textit{pairwise} statistic, dependent on the joint ranking of predicted scores across all positive-negative pairs. It is not additive over disjoint subsets of data.

\subsection{Motivation and Non-Additivity}

Suppose the dataset is partitioned into $K$ disjoint clusters $C_1, \dots, C_K$, where each cluster contains both positive and negative instances. Let $\mathcal{P}_k$ and $\mathcal{N}_k$ denote the positive and negative samples in cluster $C_k$. AUC is fundamentally a pairwise statistic—defined over all possible positive-negative pairs. Therefore, computing AUC within each cluster and taking a weighted average does not recover the global AUC.

\begin{equation}
\text{AUC} \neq \sum_{k=1}^K w_k \cdot \text{AUC}_k, \quad w_k = \frac{\left|\mathcal{P}_k\right| \cdot \left|\mathcal{N}_k\right|}{\left|\mathcal{P}\right| \cdot \left|\mathcal{N}\right|}
\end{equation}

This is because global AUC includes \textit{cross-cluster} comparisons (e.g., $x^+ \in C_i$, $x^- \in C_j$ for $i \neq j$), which are not captured in within-cluster AUCs. Ignoring these inter-cluster contributions may lead to misleading conclusions.

\subsection{Formal Decomposition Theorem}

\begin{theorem}[Cluster Decomposition of Global AUC]
\label{auc_decomp}
Let $\mathcal{P} = \bigcup_{k=1}^K \mathcal{P}_k$ and $\mathcal{N} = \bigcup_{k=1}^K \mathcal{N}_k$. Then the global AUC satisfies:
\[
\mathrm{AUC}_{\mathrm{global}} = \sum_{i=1}^{K} \sum_{j=1}^{K} w_{ij} \cdot \mathrm{AUC}_{ij}
\]
where:
\begin{align*}
\mathrm{AUC}_{ij} &= \mathbb{P}\left(\hat{f}(x^+) > \hat{f}(x^-)\ |\ x^+ \in \mathcal{P}_i,\ x^- \in \mathcal{N}_j\right) \\
w_{ij} &= \frac{\left|\mathcal{P}_i\right| \cdot \left|\mathcal{N}_j\right|}{\left|\mathcal{P}\right| \cdot \left|\mathcal{N}\right|}
\end{align*}
\end{theorem}

\begin{proof}
The AUC is defined as:
\[
\text{AUC} = \mathbb{P}\left(\hat{f}(x^+) > \hat{f}(x^-)\right), \quad x^+ \in \mathcal{P},\ x^- \in \mathcal{N}
\]
The set of all positive-negative pairs $\mathcal{P} \times \mathcal{N}$ can be partitioned by cluster membership:
\[
\mathcal{P} \times \mathcal{N} = \bigcup_{i=1}^{K} \bigcup_{j=1}^{K} \left(\mathcal{P}_i \times \mathcal{N}_j\right)
\]
The global AUC is defined as:
\[
\mathrm{AUC} = \frac{1}{\left|\mathcal{P}\right|\cdot\left|\mathcal{N}\right|} \sum_{x^+ \in \mathcal{P}} \sum_{x^- \in \mathcal{N}} \mathbf{1}\left\{\hat{f}(x^+) > \hat{f}(x^-)\right\}
\]
Partition the sums by cluster:
\[
= \frac{1}{\left|\mathcal{P}\right|\cdot\left|\mathcal{N}\right|} \sum_{i=1}^{K} \sum_{j=1}^{K} \sum_{x^+ \in \mathcal{P}_i} \sum_{x^- \in \mathcal{N}_j} \mathbf{1}\left\{\hat{f}(x^+) > \hat{f}(x^-)\right\}
\]
Define:
\[
\mathrm{AUC}_{ij} = \frac{1}{\left|\mathcal{P}_i\right|\cdot\left|\mathcal{N}_j\right|} \sum_{x^+ \in \mathcal{P}_i} \sum_{x^- \in \mathcal{N}_j} \mathbf{1}\left\{\hat{f}(x^+) > \hat{f}(x^-)\right\}
\]
and let:
\[
w_{ij} = \frac{\left|\mathcal{P}_i\right|\cdot\left|\mathcal{N}_j\right|}{\left|\mathcal{P}\right|\cdot\left|\mathcal{N}\right|}
\]
Then:
\[
\mathrm{AUC} = \sum_{i=1}^{K} \sum_{j=1}^{K} w_{ij} \cdot \mathrm{AUC}_{ij}
\]

Then the total probability can be expressed as:
\[
\text{AUC} = \sum_{i=1}^{K} \sum_{j=1}^{K} \frac{\left|\mathcal{P}_i\right| \cdot \left|\mathcal{N}_j\right|}{\left|\mathcal{P}\right| \cdot \left|\mathcal{N}\right|} \cdot \mathbb{P}\left(\hat{f}(x^+) > \hat{f}(x^-) \mid x^+ \in \mathcal{P}_i,\ x^- \in \mathcal{N}_j\right)
\]
\end{proof}

\begin{remark}
The diagonal terms $\text{AUC}_{ii}$ represent intra-cluster performance—how well the model separates positives and negatives within the same group. The off-diagonal terms $\text{AUC}_{ij}$ ($i \ne j$) reflect cross- or inter-cluster ranking alignment. Poor inter-cluster AUC values may indicate that one cluster's positives are consistently scored below another’s negatives, suggesting potential score miscalibration or bias.
\end{remark}

\subsection{Interpretation and Practical Use}

This decomposition allows us to isolate clusters where the model performs poorly either internally (intra-cluster AUC) or externally (inter- or cross-cluster ranking). Practitioners can use this framework to:
\begin{itemize}
    \item Identify localized ranking failures not visible in the global AUC,
    \item Diagnose fairness and calibration issues at the subgroup level, and
    \item Visualize performance via AUC matrices or heat maps across cluster pairs.
\end{itemize}

In the next section, we apply this decomposition to both toy data and real models trained on two case studies using the Taiwan Credit dataset and Real-Time Credit Card Fraud Detection dataset. We use these case studies to explore localized diagnostics using both the AUC and Brier score metrics.

\subsubsection{What Intra- and Inter-Cluster AUC Reveal Beyond Global AUC}

While the global AUC summarizes a model’s overall ability to rank positive examples ahead of negative ones, it compresses all pairwise information into a single scalar value. This aggregation obscures the structure of model performance across data subpopulations. The decomposition into intra- and inter-cluster AUCs exposes two critical dimensions of model behavior:

\begin{itemize}
    \item \textbf{Intra-cluster AUC} ($\text{AUC}_{ii}$): Measures the model's ability to rank within a single cluster. Low values may signal that the model is not discriminative within that group — perhaps due to weak features, poor representation, or bias.
    \item \textbf{Inter-cluster AUC} ($\text{AUC}_{ij}$, $i \ne j$): Measures how the model ranks examples across clusters. This evaluates cross-cluster calibration and relative score consistency. Low values may suggest that one cluster systematically receives higher or lower scores, regardless of true class.
\end{itemize}

\subsection{Illustrative Example: Toy Data}

To demonstrate the decomposition, consider a toy dataset with 6 samples divided into two clusters:

\begin{center}
\begin{tabular}{ccc}
\toprule
Cluster & Label & Predicted Score \\
\midrule
$C_1$ & 1 & 0.9 \\
$C_1$ & 1 & 0.8 \\
$C_1$ & 0 & 0.4 \\
$C_2$ & 1 & 0.6 \\
$C_2$ & 0 & 0.7 \\
$C_2$ & 0 & 0.3 \\
\bottomrule
\end{tabular}
\end{center}

There are 3 positives and 3 negatives overall. Compute:

\[
\begin{aligned}
w_{11} &= \frac{2 \cdot 1}{3 \cdot 3} = \frac{2}{9}, \quad
w_{12} = \frac{2 \cdot 2}{3 \cdot 3} = \frac{4}{9}, \\
w_{21} &= \frac{1 \cdot 1}{3 \cdot 3} = \frac{1}{9}, \quad
w_{22} = \frac{1 \cdot 2}{3 \cdot 3} = \frac{2}{9}
\end{aligned}
\]

\textbf{AUC Contributions}:
\[
\begin{aligned}
\mathrm{AUC}_{11} & = 1.0 \quad \text{(both $C_1$ positives $> C_1$ negative)} \\
\mathrm{AUC}_{12} & = 1.0 \quad \text{($C_1$ positives $> C_2$ negatives)} \\
\mathrm{AUC}_{21} & = 1.0 \quad \text{($C_2$ positive $> C_1$ negative)} \\
\mathrm{AUC}_{22} & = 0.5 \quad \text{($C_2$ positive $> 1$ of 2 $C_2$ negatives)} \\
\end{aligned}
\]

\textbf{Global AUC}:
\[
\mathrm{AUC} = \frac{2}{9}\cdot(1.0) + \frac{4}{9}\cdot(1.0) + \frac{1}{9}\cdot(1.0) + \frac{2}{9}\cdot(0.5) = \frac{8}{9} \approx 0.889
\]

\paragraph{Interpretation} While the global AUC appears high (0.889), the poor within-cluster performance of $C_2$ (\(\mathrm{AUC}_{22} = 0.5\)) indicates a local failure mode. This would be missed by global metrics but exposed by the decomposition.

This example illustrates the practical utility of the decomposition: understanding performance not just globally, but in terms of interactions between subpopulations—critical in regulated domains like credit scoring and fraud detection.

\subsection{Illustrative Example: Credit Risk Model Across Demographic Segments}

Suppose a credit risk model is applied to three clusters:
\begin{itemize}
    \item $C_1$: Young adults with limited credit history
    \item $C_2$: Middle-aged individuals with stable income
    \item $C_3$: Seniors with irregular cash flows
\end{itemize}

A global AUC of 0.90 might suggest strong performance, but decomposition reveals:

\[
\text{AUC}_{\text{global}} = 0.90 \quad \text{(aggregated over all clusters)}
\]

\[
\text{AUC matrix:}
\quad
\begin{bmatrix}
\text{AUC}_{11} & \text{AUC}_{12} & \text{AUC}_{13} \\
\text{AUC}_{21} & \text{AUC}_{22} & \text{AUC}_{23} \\
\text{AUC}_{31} & \text{AUC}_{32} & \text{AUC}_{33}
\end{bmatrix}
=
\begin{bmatrix}
0.62 & 0.81 & 0.79 \\
0.85 & 0.92 & 0.90 \\
0.60 & 0.84 & 0.70
\end{bmatrix}
\]

From this, we observe:
\begin{itemize}
    \item Poor intra-cluster AUC for $C_1$ and $C_3$ ($\text{AUC}_{11}=0.62$, $\text{AUC}_{33}=0.70$): suggests the model struggles to rank risk accurately within these subgroups.
    \item Low inter-cluster AUCs ($\text{AUC}_{31} = 0.60$): indicates that seniors ($C_3$) are often scored lower than young adult positives from $C_1$, potentially leading to unfair denial of credit.
    \item High intra-cluster and inter-cluster AUC for $C_2$: the model performs well for middle-aged individuals.
\end{itemize}

\paragraph{Actionable Insights}

This decomposition allows stakeholders to:
\begin{itemize}
    \item Identify subgroups with unreliable internal rankings (potential for feature engineering or retraining),
    \item Detect relative miscalibration across clusters (potential for \textit{post-hoc} score adjustment or fairness remediation), and
    \item Quantify the contribution of each cluster-pair to the global AUC.
\end{itemize}

\section{Additive Alternatives: Brier Score and Log Loss}

Unlike the AUC, both the Brier score and log loss are additive metrics that can be directly decomposed across clusters. A comprehensive comparison of AUC, log loss, and Brier score for probabilistic models was presented by Niculescu-Mizil and Caruana \cite{niculescu2005predicting}.

\subsection{Brier Score}
The Brier score, originally introduced for weather forecasting \cite{brier1950verification}, measures the squared error of predicted probabilities:

\begin{equation}
\text{Brier score} = \frac{1}{n} \sum_{i=1}^{n} \left(\hat{p}_i - y_i\right)^2
\end{equation}

It decomposes naturally by cluster:
\begin{equation}
\text{Brier}_{\text{global}} = \sum_{k=1}^{K} \frac{n_k}{n} \cdot \text{Brier}_{k}
\end{equation}

\subsection{Log Loss}

log loss (cross-entropy loss) measures the expected log-likelihood of the predicted probabilities:

\begin{equation}
\text{LogLoss} = -\frac{1}{n} \sum_{i=1}^{n} \left[y_i \log\left(\hat{p}_i\right) + \left(1 - y_i\right) \log\left(1 - \hat{p}_i\right)\right]
\end{equation}

It also decomposes across clusters:
\begin{equation}
\text{LogLoss}_{\text{global}} = \sum_{k=1}^{K} \frac{n_k}{n} \cdot \text{LogLoss}_k
\end{equation}

\subsection{Use Case Comparison: When to Use AUC, Brier Score, and Log Loss Together}

Each of the metrics—AUC, Brier score, and log loss—offers a distinct lens into model performance. Understanding how to interpret and use them in tandem can significantly enhance model diagnostics, particularly in clustered or segmented data scenarios.

\paragraph{AUC: Rank-Based Discrimination}

The AUC measures the model’s ability to assign higher scores to positives than negatives, regardless of absolute probability calibration. It is most appropriate when:
\begin{itemize}
    \item The model is used for ranking or prioritization (e.g., flagging high-risk credit applicants for further review).
    \item Thresholds will be selected later or vary by context.
    \item You want to assess separability, even when predicted probabilities are not well calibrated.
\end{itemize}
However, the AUC does not capture how close predicted scores are to true probabilities, and it is not additive—thus it hides local failures if the global score is high.

\paragraph{Brier Score: Calibration + Discrimination}

The Brier score penalizes both poor calibration and poor separation. It is useful when:
\begin{itemize}
    \item You want to assess the overall probabilistic accuracy of the model.
    \item Well-calibrated probabilities are required for decision-making (e.g., setting pricing, expected risk).
    \item You want a decomposable metric to monitor cluster-wise reliability.
\end{itemize}
Because it is quadratic, the Brier score is sensitive to how far predicted probabilities deviate from the truth, even when ranking is correct.

\paragraph{Log Loss: Likelihood-Sensitive Evaluation}

Log loss (or cross-entropy loss) emphasizes confidence in correct predictions. It is useful when:
\begin{itemize}
    \item You need to maximize the log-likelihood of the true labels under the predicted probabilities.
    \item You want to penalize overconfident but incorrect predictions heavily.
    \item The model is being used in downstream probabilistic decision-making (e.g., Bayesian pipelines, customer loan-to-value or LTV estimation).
\end{itemize}
Log loss, like the Brier score, is additive and decomposable by cluster. It is more aggressive than the Brier score in penalizing miscalibrated or overconfident predictions.

\paragraph{Using All Three Together}

Using AUC, Brier score, and log loss jointly provides a holistic view of model performance:
\begin{itemize}
    \item \textbf{AUC high, but Brier score/log loss poor:} The model ranks correctly but is poorly calibrated. Probabilities may need recalibration.
    \item \textbf{AUC low, Brier score/log loss moderate:} The model is well-calibrated but unable to discriminate between classes. This may indicate weak features or overlapping distributions. 
    \item \textbf{AUC and Brier score/log loss high in one cluster, low in another:} Reveals localized failure—global metrics are masking cluster-level bias or risk.
\end{itemize}

For cluster-level diagnostics:
\begin{itemize}
    \item Use AUC$_{ij}$ (from our decomposition) to assess whether the model separates positives from negatives both \textit{within} and \textit{across} clusters.
    \item Use the Brier score and log loss per cluster to measure calibration and confidence in each subgroup.
    \item Investigate clusters where the Brier score and log loss diverge from the AUC to identify mismatch between ranking ability and confidence.
\end{itemize}

\begin{table}[htbp]
\centering
\begin{tabular}{@{}lccc@{}}
\toprule
Metric        & Additive? & Interpretable? & Pairwise? \\
\midrule
AUC           & No         & Yes (ranking)     & Yes         \\
Brier score   & Yes         & Yes (calibration) & No         \\
Log Loss      & Yes         & Yes (likelihood)  & No         \\
\bottomrule
\end{tabular}
\caption{Comparison of performance metrics for binary classification}
\end{table}

The AUC is valuable for assessing ranking quality, while the Brier score and log loss provide additive diagnostics useful for calibration, segment-level reliability, and model debugging.

\section{Case Study}
\subsection{Case Study 1: Taiwan Credit Risk Dataset}

To demonstrate the practical relevance of our decomposition framework and compare models in a real-world setting, we use the \textit{Taiwan Credit Card Default Prediction Dataset}, a publicly available benchmark from the UCI Machine Learning Repository \cite{yeh2009comparisons}.

\subsubsection{Dataset Description}

The Taiwan Credit dataset contains records for 30,000 credit card clients in Taiwan, including demographic features (e.g., age, gender, education), payment history (e.g., bill amounts and repayments over six months), and a binary target indicating default in the next month. The goal is to build a classifier to predict whether a customer will default. This dataset is frequently used in credit risk modeling and benchmarking supervised learning algorithms.

\subsubsection{Models and Experimental Setup}

We evaluate three tree-based gradient boosting methods:
\begin{itemize}
    \item \textbf{LightGBM (LGBM)} \cite{ke2017lightgbm}
    \item \textbf{XGBoost (XGB)} \cite{chen2016xgboost}
    \item \textbf{CatBoost (CBoost)} \cite{dorogush2018catboost}
\end{itemize}

To focus on simple and interpretable models, we configure all algorithms to use \textbf{decision stumps} (i.e., tree depth = 1), using their default hyperparameters otherwise. Each model is trained on 80\% of the data and evaluated on the remaining 20\% test set. We report AUC, Brier score, and log loss for training and test data, and compute the generalization gap (train minus test) for each metric.

\subsubsection{Gradient Boosting with Decision Stumps}

Gradient Boosted Decision Trees (GBDTs) are ensemble methods that sequentially train weak learners (e.g., decision trees) to fit the negative gradient of a loss function. When the base learner is a decision stump (i.e., a tree of depth 1), the model becomes an additive model of simple rules based on single features. This setup resembles a generalized additive model (GAM), where each feature contributes independently to the final prediction \cite{sudjianto2024gamb}. Despite their simplicity, such models can be surprisingly effective and are highly interpretable.

\subsubsection{Model Performance Comparison and Discussion}

\begin{table}[ht]
\centering
\caption{Comparison of Metrics (AUC, Brier Score, Log Loss) Across Models}
\label{tab:model_comparison}
\begin{tabular}{l|ccc|ccc|ccc}
\toprule
\textbf{Set} & \multicolumn{3}{c|}{\textbf{LGBM}} & \multicolumn{3}{c|}{\textbf{XGB}} & \multicolumn{3}{c}{\textbf{CBoost}} \\
& AUC & Brier & LogLoss & AUC & Brier & LogLoss & AUC & Brier & LogLoss \\
\midrule
Train & 0.7759 & 0.1372 & 0.4372 & 0.7843 & 0.1355 & 0.4310 & 0.7851 & 0.1353 & 0.4305 \\
Test  & 0.7744 & 0.1328 & 0.4274 & 0.7792 & 0.1320 & 0.4237 & 0.7799 & 0.1319 & 0.4234 \\
GAP   & -0.0015 & -0.0044 & -0.0098 & -0.0051 & -0.0035 & -0.0073 & -0.0052 & -0.0034 & -0.0071 \\
\bottomrule
\end{tabular}
\end{table}

All three models achieve similar performance, with CatBoost slightly outperforming others in Brier score and log loss. The generalization gap is small across metrics, suggesting limited overfitting under the decision stump constraint. While the AUC remains relatively stable between train and test, log loss and the Brier score provide finer discrimination and more sensitivity to probability calibration.

This example illustrates how multiple metrics—especially additive ones like the Brier score and log loss—can reveal nuances in model generalization not visible from the AUC alone. This motivates our proposal to combine the global AUC with decomposed and localized metrics for robust model evaluation.

\subsubsection{Cluster-Level Diagnostic Insights from the AUC and Brier Score Visualizations}

Figures~\ref{fig:cluster_auc} and~\ref{fig:cluster_brier} visualize the per-cluster performance of all models, illustrating the complementary strengths of AUC and Brier score for diagnostic analysis. 

Figure~\ref{fig:cluster_auc} shows that while many clusters achieve relatively high AUCs (indicating strong within-cluster rank separation), certain clusters deviate noticeably. For example, a few clusters exhibit AUC values near or below 0.6—suggesting that within those subpopulations, the model struggles to rank positives above negatives. These local failures are easily overlooked when relying on the global AUC alone.

In contrast, Figure~\ref{fig:cluster_brier} reveals variability in probabilistic accuracy across clusters. Some clusters with reasonably high AUCs still have high Brier scores, indicating poor calibration—i.e., the model assigns extreme probabilities that do not align well with actual outcomes. Conversely, clusters with moderate AUCs but low Brier scores may be well-calibrated but lack discriminative features.

Together, these figures underscore the value of using both decomposition-based rank metrics (e.g., AUC) and additive error metrics (e.g., Brier score) to gain a fuller understanding of model behavior across subgroups. In practice, this enables targeted improvements—such as re-engineering features or applying post-processing calibration—to specific clusters without affecting global performance.

The cluster definitions and corresponding performance metrics shown in Figures~\ref{fig:cluster_auc} and~\ref{fig:cluster_brier} were generated using \textbf{MoDeVa.ai}—a model diagnostics and validation platform designed for interpretable and localized performance analysis. MoDeVa's cluster-based evaluation tools allow users to partition the input space into behaviorally coherent regions and compute decomposed metrics, such as intra-cluster AUC and per-cluster Brier scores, with minimal effort. For details, see the official user guide \cite{modeva2024}.

    \begin{figure}[htbp]
    \centering
    \begin{minipage}{0.63\textwidth}
        \centering
        \includegraphics[width=\textwidth]{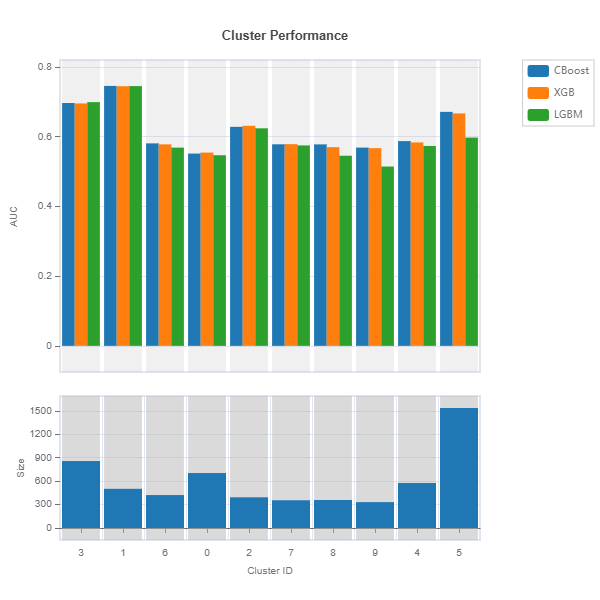}
        \caption{Clusters of inputs and their AUCs}
        \label{fig:cluster_auc}
    \end{minipage}\vfill
    \begin{minipage}{0.63\textwidth}
        \centering
        \includegraphics[width=\textwidth]{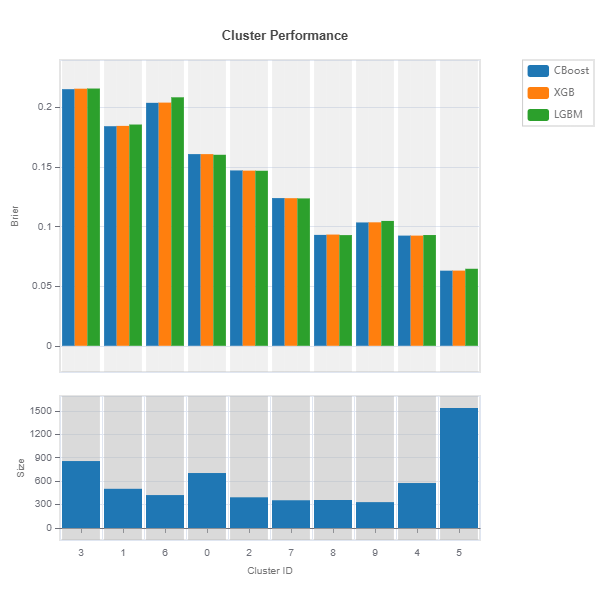}
        \caption{Clusters of inputs and their Brier scores}
        \label{fig:cluster_brier}
    \end{minipage}
    \end{figure}

\subsection{Case Study 2: Real-Time Credit Card Fraud Detection}
\label{sec:fraud_case}
Credit card fraud modeling is an archetypal \emph{needle-in-a-haystack} problem: the class of interest (fraud) occurs in fewer than \(2\%\) of transactions, yet misclassification can lead to direct financial loss or customer frustration from false declines. Moreover, fraud patterns evolve quickly, today’s high-risk merchant may be benign tomorrow. Raw features such as \texttt{merchant} and \texttt{mcc} (merchant category code) exhibit extreme cardinality. These characteristics make the domain ideal for demonstrating how cluster-level AUC and additive metrics surface model weaknesses that a single global statistic obscures.

\subsubsection{Data and Feature Engineering}
The dataset comprises
\(200{,}000\) synthetically generated transaction authorization records, emulating data collected by financial institutions. Each record contains the transformed transaction features for the model, as noted in Table \ref{tab:feat_eng}, along with binary outcome \texttt{fraud\_status}. Fraud prevalence is \(1.83\%\), indicating the highly imbalanced nature of this data.

The model relies on a compact set of \textbf{14} real-time features
(Table \ref{tab:feat_eng}), designed to capture temporal rhythm,
behavioral velocity, account and channel (i.e., card-present or card-not-present) status, and multi-granular merchant risk.

\begin{enumerate}
  \item \emph{Temporal rhythm.}  
        Hour of the day, day of the week, weekend flag, and a night-hours flag
        (\(00{:}00\) – \(05{:}59\)).
  \item \emph{Behavioral velocity.}  
        Seconds since the customer’s previous transaction, count of prior
        transactions in the past hour, and total spend in the past
        24 hours.
  \item \emph{Account \& channel status.}  
        Utilization ratio, account age, a \texttt{cvv\_match} indicator,
        and a Boolean \texttt{card\_present} flag
        (1 = card-present, 0 = card-not-present).
  \item \emph{Multi-granular merchant risk.}  
        Two separate Bayesian-smoothed fraud rates at two different levels of merchant granularity:  
        \begin{itemize}
            \item \texttt{merchant\_name\_risk} — learned at the exact
                  \texttt{merchant} level; and
            \item \texttt{mcc\_risk} — learned at the
                  \texttt{mcc} level. 
        \end{itemize}
        Using both levels of granularity lets the model exploit
        merchant-specific histories while still assigning a sensible
        prior to unfamiliar merchants that share an MCC with observed
        fraud.
\end{enumerate}

All fourteen features are computable with sub-millisecond latency and
require no external credit bureau data.

\begin{table}[ht]
\centering
\caption{Engineered real-time features used by the fraud model.}
\label{tab:feat_eng}
\begin{tabular}{l l}
\toprule
\textbf{Feature} & \textbf{Description} \\
\midrule
\(\mathsf{hour}\)                    & Hour of day (0–23) \\
\(\mathsf{day\_of\_week}\)           & Monday (0) … Sunday (6) \\
\(\mathsf{is\_weekend}\)             & 1 if Saturday or Sunday \\
\(\mathsf{is\_night}\)               & 1 if \(00{:}00\!\le t<\!06{:}00\) \\
\(\mathsf{secs\_since\_prev}\)       & Seconds since previous transaction (customer) \\
\(\mathsf{txns\_last\_1h}\)          & Count of prior transactions in past 1 hour \\
\(\mathsf{amt\_sum\_24h}\)           & Spend (USD) in past 24 hours (before transaction) \\
\(\mathsf{utilization}\)             & \(\texttt{transaction\_amount}/\texttt{credit\_limit}\) \\
\(\mathsf{acct\_open\_days}\)        & Days since account opening \\
\(\mathsf{cvv\_match}\)              & 1 if entered CVV matches stored CVV \\
\({\mathsf{card\_present}}\)          & 1 = card-present; 0 = card-not-present \\
\({\mathsf{merchant\_name\_risk}}\)   & Fraud-rate score at \texttt{merchant} granularity \\
\({\mathsf{mcc\_risk}}\)             & Fraud-rate score at \texttt{mcc} granularity \\
\(\mathsf{fraud\_status}\)            & Binary target label (kept for training) \\
\bottomrule
\end{tabular}
\end{table}

\subsubsection{Model and Evaluation}
Because the transaction stream contains two high-cardinality categorical variables such as \texttt{merchant} \((\sim\!2{,}400\) unique values) and \texttt{mcc}, we selected \textbf{CatBoost}, which handles categorical features via ordered target encoding and avoids the memory blow-up of one-hot representations.

After a coarse grid search we fixed the hyperparameters at

\begin{center}
\texttt{learning\_rate = 0.05,\;
max\_depth = 3,\;
n\_estimators = 500,\;
scale\_pos\_weight = 53.9639},
\end{center}

where the class weight (i.e., used for \texttt{scale\_pos\_weight}) equals the training-fold ratio of non-fraud to fraud observations. All other settings apply the default CatBoost parameters. The dataset was split chronologically: the first \(80\,\%\) of transactions for training and the final \(20\,\%\) for hold-out testing, preventing temporal leakage.

\begin{table}[ht]
\centering
\caption{CatBoost performance on the credit card fraud dataset
         (\(80{:}20\) time-based split).}
\label{tab:fraud_catboost_perf}
\begin{tabular}{lccc}
\toprule
\textbf{Split} & \textbf{AUC} & \textbf{Log Loss} & \textbf{Brier} \\
\midrule
Train & 0.7876 & 0.5659 & 0.2023 \\
Test  & 0.7713 & 0.5678 & 0.2032 \\
\midrule
Generalization gap & $-0.0163$ & $+0.0018$ & $+0.0009$ \\
\bottomrule
\end{tabular}
\end{table}

The CatBoost model achieves an AUC of \(0.771\) on the unseen test data while showing virtually no degradation in log loss or Brier scores between the training and testing data. The small generalization gaps among these three performance metrics provide evidence of regularization against over-fitting in a highly imbalanced setting.

Following decomposition outlined in Theorem \ref{auc_decomp}, the global total AUC of \(0.771\) decomposes to global intra- and inter-cluster AUC totals of \(0.120\) and \(0.651\), which can be used to provide cluster-level diagnostic insights.

\subsubsection{Cluster-Level Diagnostic Insights}
\label{sec:fraud_diag}

Figures \ref{fig:cluster_fraud_auc} and \ref{fig:cluster_fraud_brier} visualize the CatBoost model through the lens of our decomposition framework, further demonstrating the corresponding holistic use of the AUC and Brier score to identify critical areas within the data to leverage for model improvement and monitoring. This type of resilience analysis identifies variables that may be susceptible to drifts in distribution, motivating awareness of model performance. While not shown, log loss demonstrates similar results and insights provided by the Brier score. As with the Taiwan Credit dataset, clusters are generated using \textbf{MoDeVa.ai} and evaluated using MoDeVa's corresponding cluster analysis tools \cite{modeva2024}.

\paragraph{Intra- vs. inter-cluster ranking (Figures \ref{fig:cluster_fraud_auc} and \ref{fig:cluster_fraud_brier}).}
Figure \ref{fig:cluster_fraud_auc} shows that a few clusters achieve AUC values near or above 0.80, yet three clusters dip below 0.60, signaling weak intra- or within-cluster discrimination for those customer–merchant segments (see Figure \ref{fig:Density_auc}). The remaining clusters have more moderate AUCs that range between 0.6 to 0.8. The global test AUC of 0.771 masks these local defects and the variation in local performance, belying the dips in performance. Note two clusters have no calculated AUC, as these clusters only include non-fraud observations.

\paragraph{Calibration disparities (Figures \ref{fig:cluster_fraud_auc} and \ref{fig:cluster_fraud_brier}).}
Cluster 5 exhibits a respectable AUC but the second highest Brier score, indicating overconfident probability estimates and poor calibration. Conversely, cluster 1 exhibits the second lowest AUC as well as the second lowest Brier score, indicating a well-calibrated model is unable to discriminate between classes to a similar degree. Such divergence between rank quality and calibration would remain hidden without the additive view of the Brier score, as the global Brier score of 0.203 obscures the inconsistencies in calibration among clusters.

    \begin{figure}[htbp]
    \centering
    \begin{minipage}{0.48\textwidth}
        \centering
        \includegraphics[width=\textwidth]{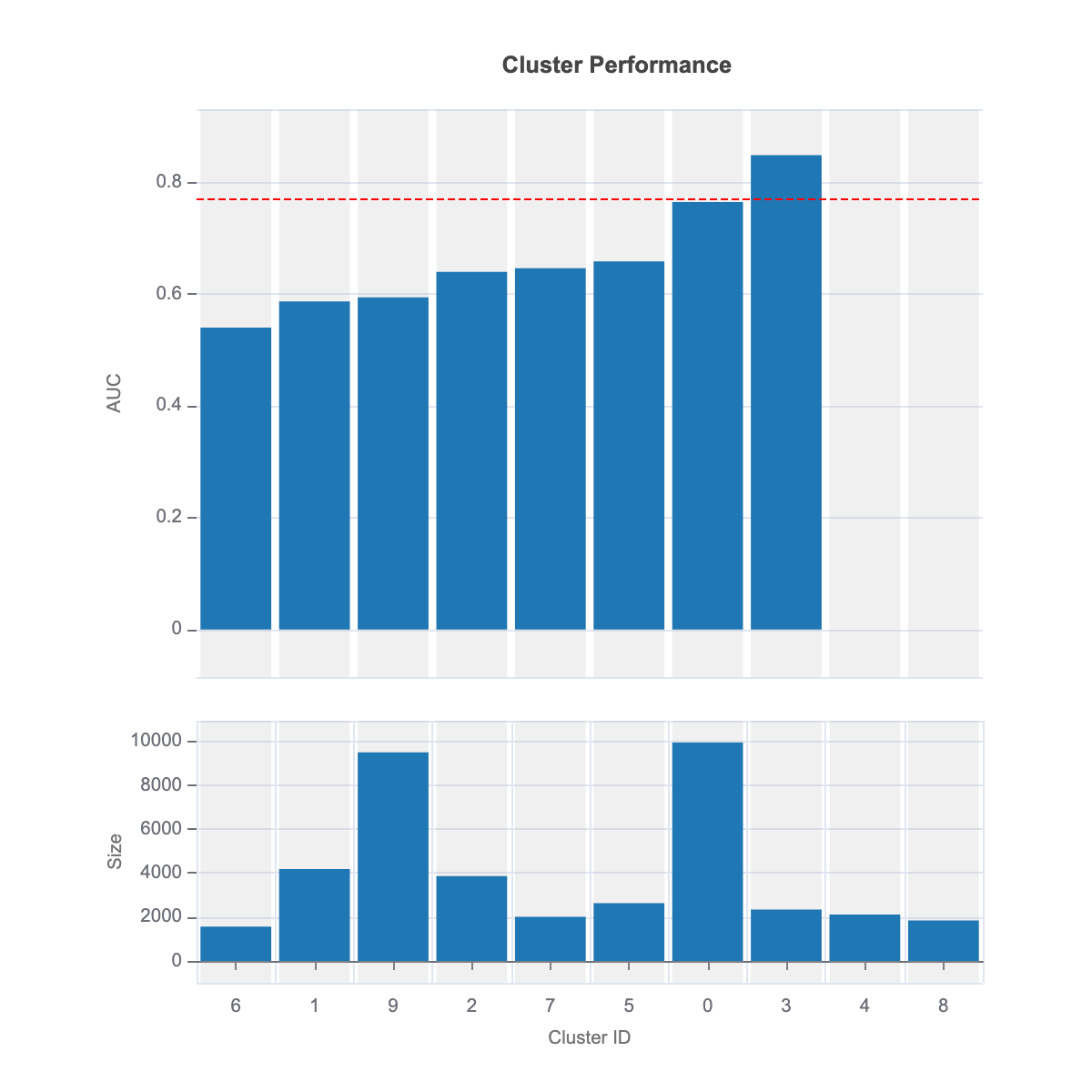}
        \caption{Clusters of inputs CatBoost fraud model and their AUCs.}
        \label{fig:cluster_fraud_auc}
    \end{minipage}\hfill
    \begin{minipage}{0.48\textwidth}
        \centering
        \includegraphics[width=\textwidth]{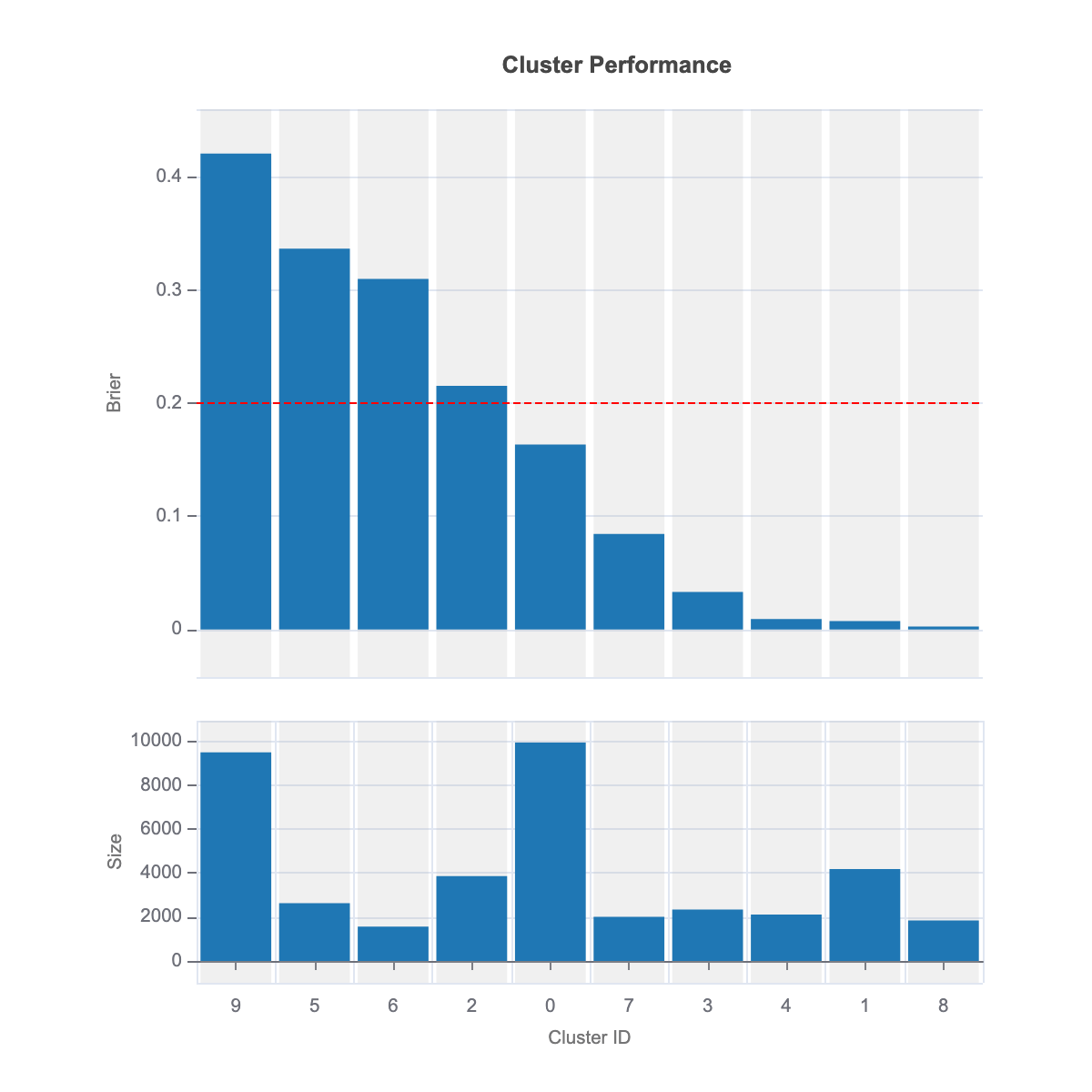}
        \caption{Clusters of inputs CatBoost fraud model and their Brier scores.}
        \label{fig:cluster_fraud_brier}
    \end{minipage}
    \end{figure}

    \begin{figure}[htbp]
    \centering
    \begin{minipage}{0.48\textwidth}
        \centering
        \includegraphics[width=\textwidth]{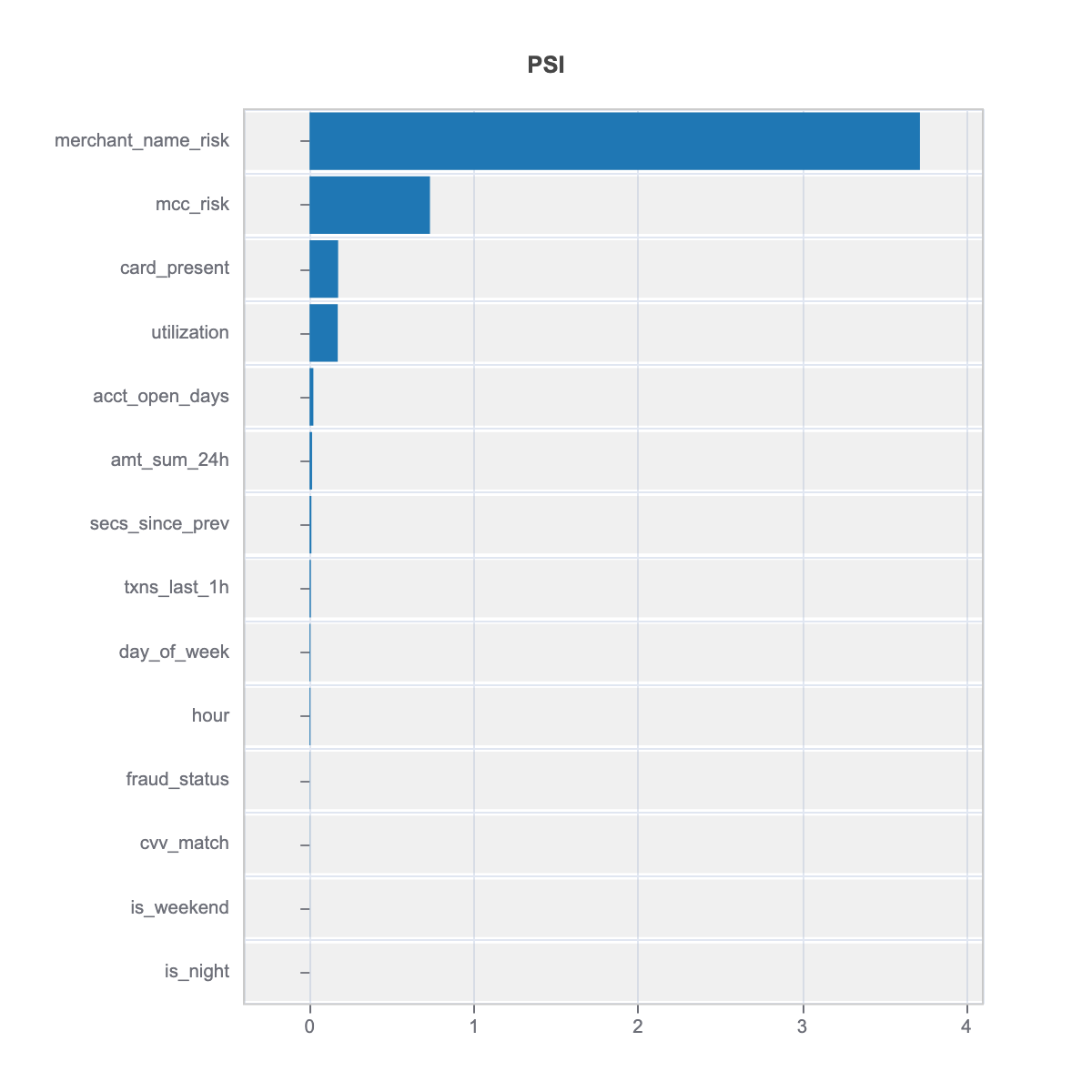}
        \caption{Feature marginal distribution distance of worst cluster to the rest of the clusters with respect to AUCs.}
        \label{fig:PSI_auc}
    \end{minipage}\hfill
    \begin{minipage}{0.48\textwidth}
        \centering
        \includegraphics[width=\textwidth]{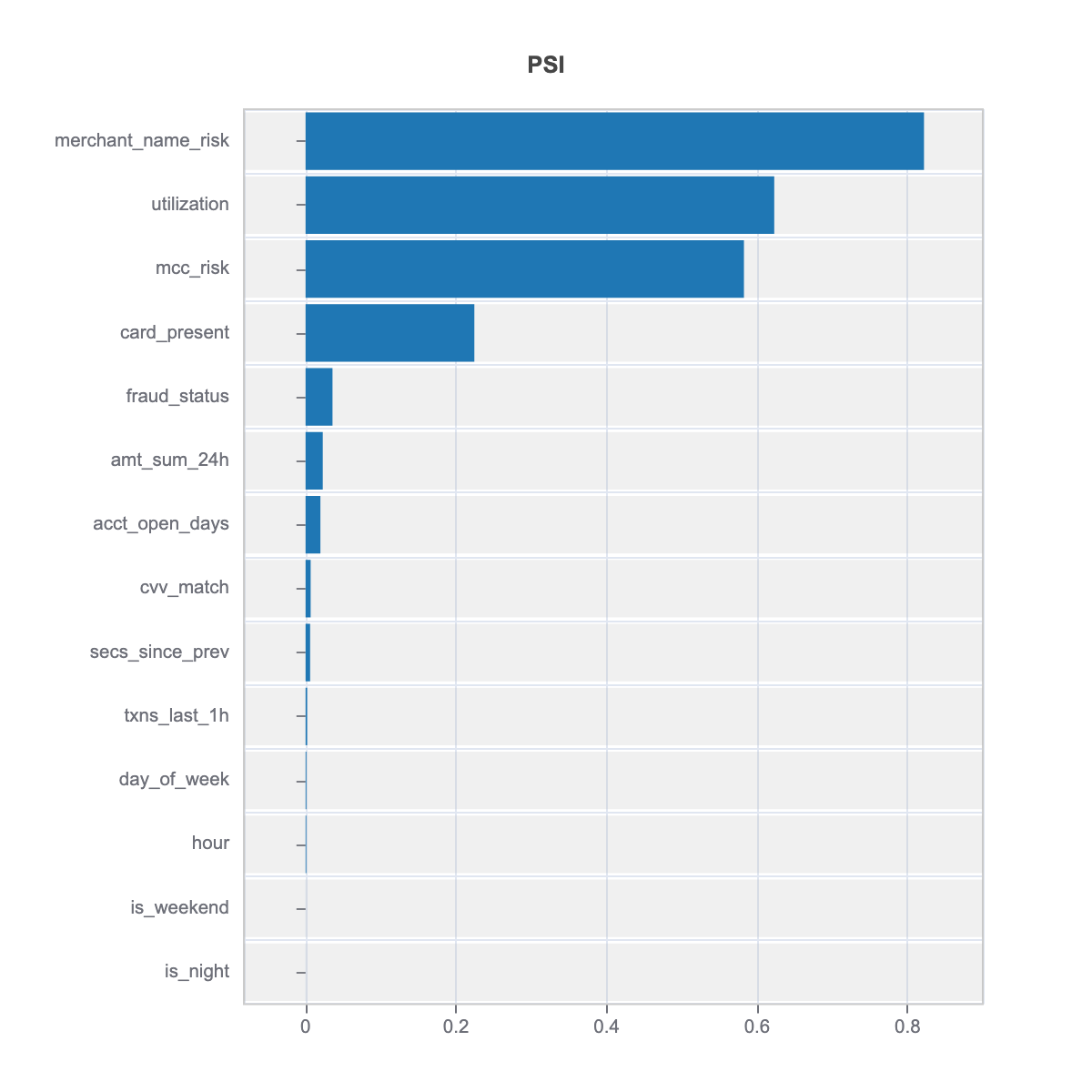}
        \caption{Feature marginal distribution distance of worst cluster to the rest of the clusters with respect to Brier scores.}
        \label{fig:PSI_brier}
    \end{minipage}
    \end{figure}

    \begin{figure}[htbp]
    \centering
    \begin{minipage}{0.48\textwidth}
        \centering
        \includegraphics[width=\textwidth]{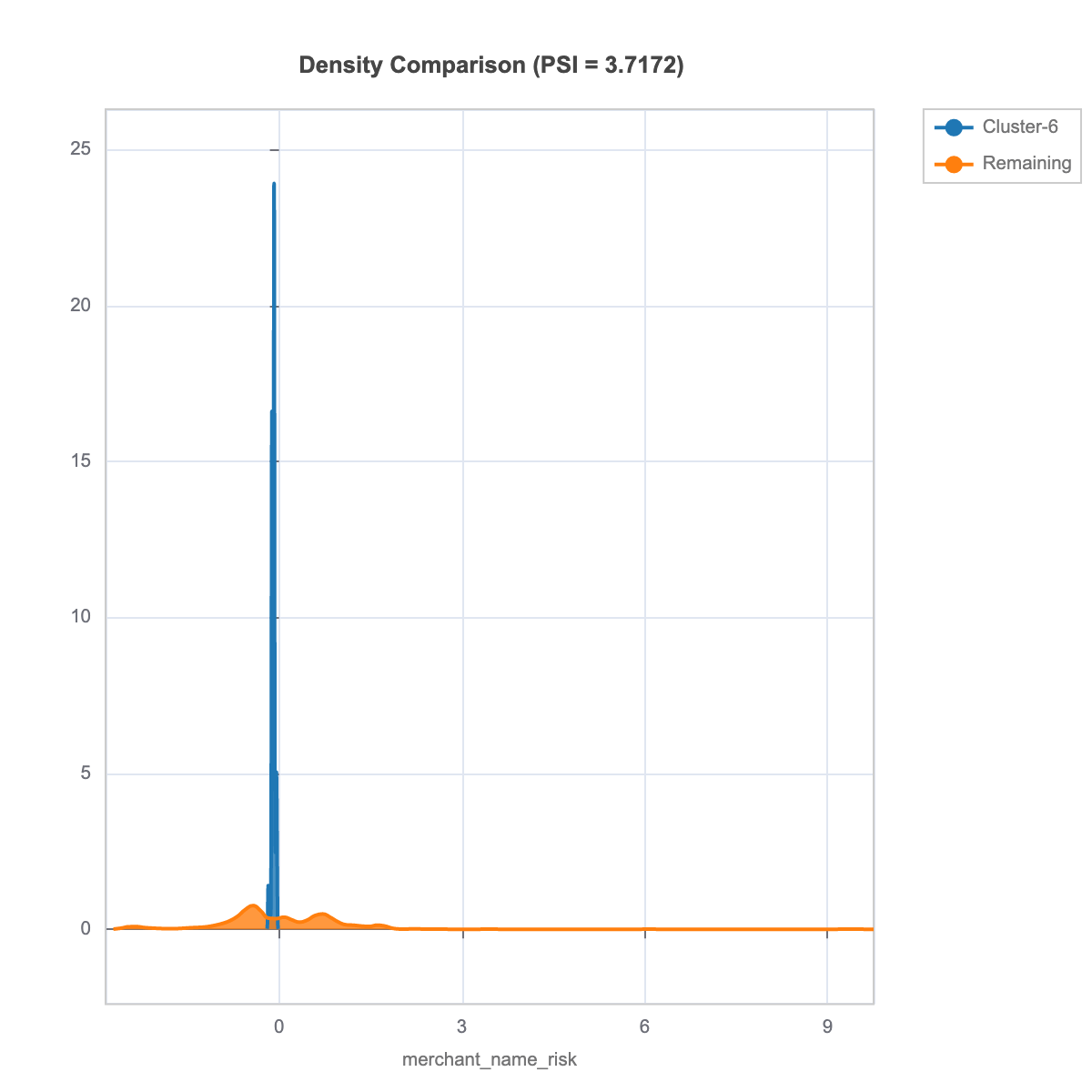}
        \caption{Feature marginal distribution distance of worst cluster to the rest of the clusters with respect to AUCs.}
        \label{fig:Density_auc}
    \end{minipage}\hfill
    \begin{minipage}{0.48\textwidth}
        \centering
        \includegraphics[width=\textwidth]{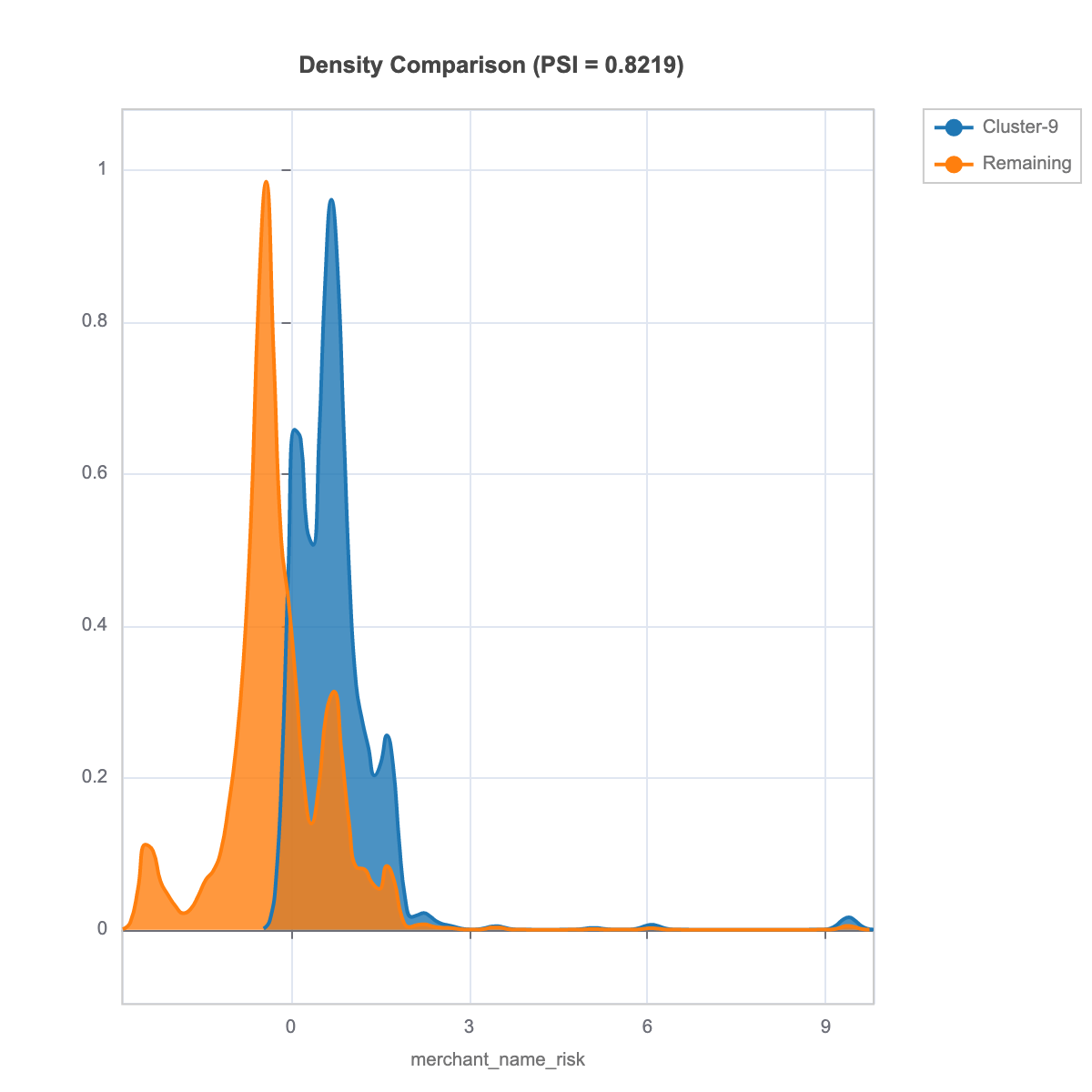}
        \caption{Feature marginal distribution distance of worst cluster to the rest of the clusters with respect to Brier scores}
        \label{fig:Density_brier}
    \end{minipage}
    \end{figure}

\paragraph{Explaining the worst cluster (Figures \ref{fig:PSI_auc} and \ref{fig:PSI_brier}).}
Figures \ref{fig:PSI_auc} and \ref{fig:PSI_brier} compare
feature distributions of the \emph{worst} AUC and Brier score clusters against the rest of the data. The largest Jensen-Shannon (J-S) distances (a smoothed and symmetrized version of the Kullback-Leibler distance), also known as the Population Stability Index (PSI) \cite{yurdakul2007analyzing}, primarily arise from \texttt{merchant\_name\_ risk}. Other features with large PSIs include \texttt{mcc\_risk}, \texttt{card\_present}, and \texttt{utilization}. Large PSIs suggest that fraud patterns for these segments differ from the rest of the clusters, as there is an observed change or difference in the characteristics among the compared populations (i.e., clusters in this case); additionally, their fraud rates are also higher (Figures \ref{fig:Density_Fraud_AUC} and \ref{fig:Density_Fraud_Brier}). This diagnostic directs model developers for model improvement and model users to monitor important features where their drift will impact performance significantly.

    \begin{figure}[htbp]
    \centering
    \begin{minipage}{0.48\textwidth}
        \centering
        \includegraphics[width=\textwidth]{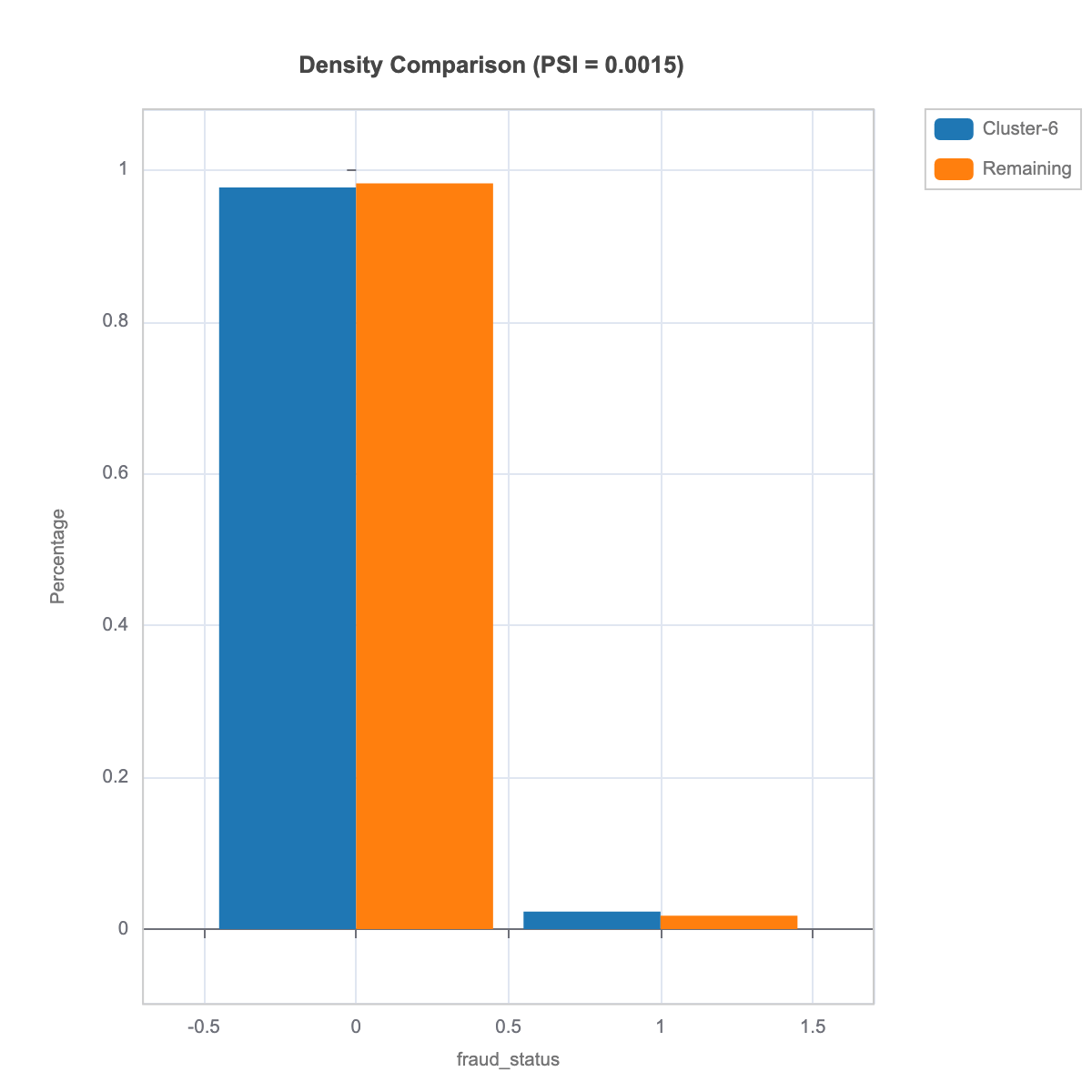}
        \caption{Fraud rate difference worst AUC cluster to the rest of the clusters.}
        \label{fig:Density_Fraud_AUC}
    \end{minipage}\hfill
    \begin{minipage}{0.48\textwidth}
        \centering
        \includegraphics[width=\textwidth]{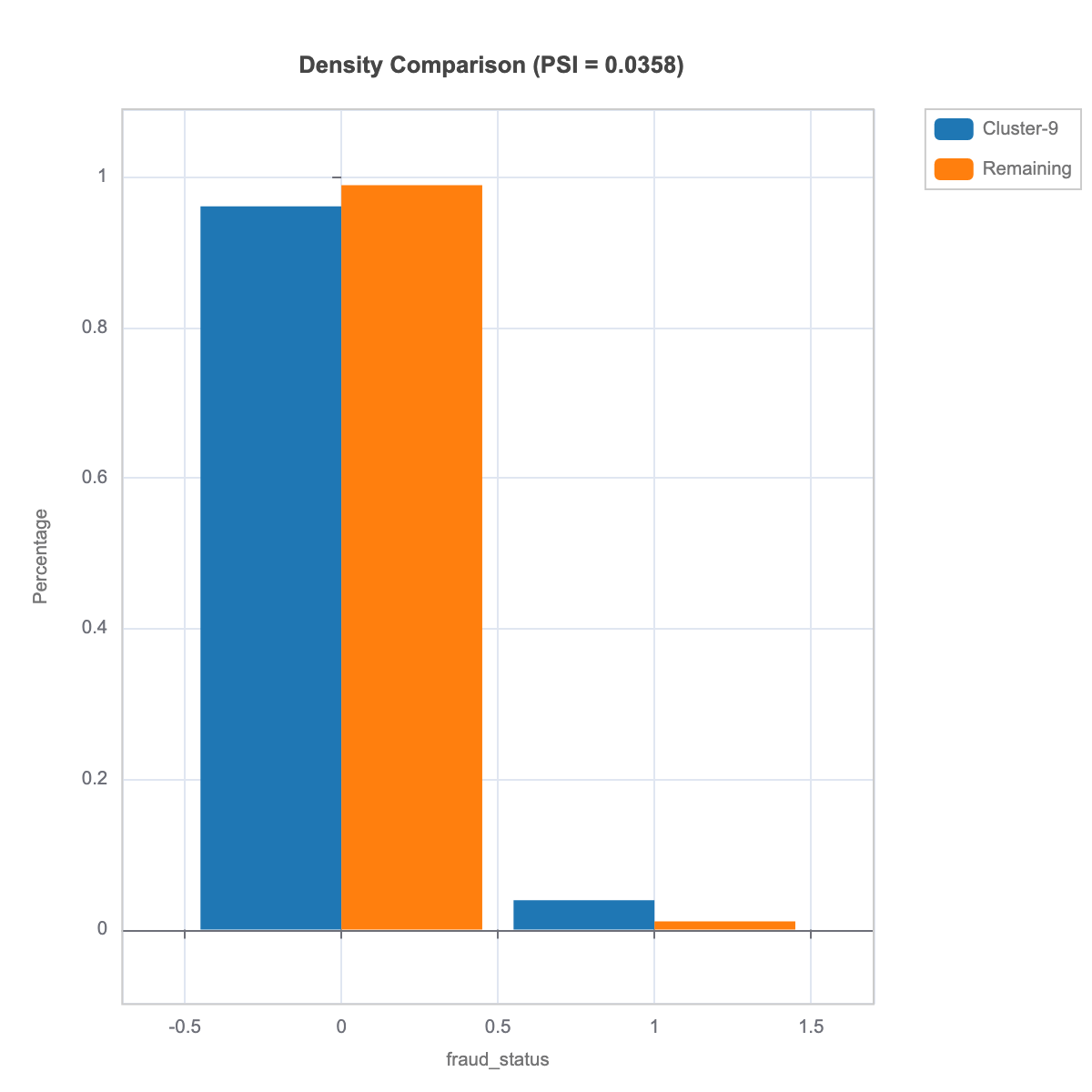}
        \caption{Fraud rate difference worst Brier score cluster to the rest of the clusters.}
        \label{fig:Density_Fraud_Brier}
    \end{minipage}
    \end{figure}

\paragraph{Evaluation of inter-cluster performance (Figures \ref{fig:Intercluster_AUC}.}
Figures \ref{fig:Intercluster_AUC} visualizes the inter-cluster AUCs as a heatmap, which provides insights to how the model ranks observations across clusters. Low values across the heatmap indicate areas of inconsistent calibration, further supporting the findings and insights the Brier scores previously provided. Recall cluster 9 had the highest Brier score performance, suggesting poor calibration, and is supported by the inter-cluster AUCs shown in \ref{fig:Intercluster_AUC}.

    \begin{figure}[htbp]
    \centering
    \begin{minipage}{0.98\textwidth}
        \centering
        \includegraphics[width=\textwidth]{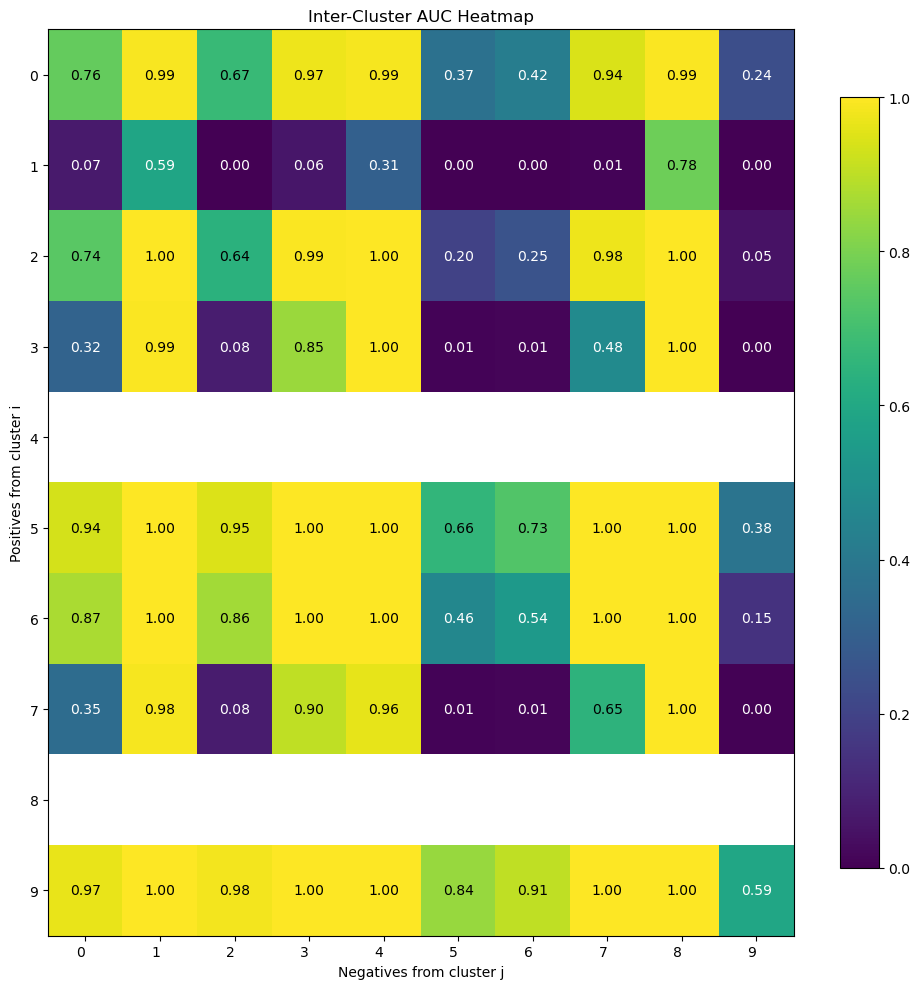}
        \caption{Heatmap of inter-cluster AUC.}
        \label{fig:Intercluster_AUC}
    \end{minipage}\hfill
    \end{figure}

\section{Conclusion}

Traditional one-number metrics mask the nuanced ways in which
classification models can fail. By formally decomposing global AUC into intra- and inter-cluster terms, we provide a mathematically sound, computationally light, and visually intuitive tool for localizing model weaknesses. When paired with additive metrics such as the Brier score and log loss, the method separates \emph{where} a model ranks poorly from \emph{where} it is miscalibrated—information that is indispensable for model risk managers, auditors, and fairness officers.

Our two case studies demonstrate the practical value of the approach. In the Taiwan Credit dataset, cluster diagnostics revealed borrowers whose scores were well-calibrated yet poorly ranked. In the CatBoost fraud detection model, the decomposition exposed small but financially critical merchant segments because they have a higher fraud rate than the rest of the population. This finding provides insight for model improvement and implementation of important features for drift monitoring during model production.

We believe the proposed framework advances the state of
model validation by making local performance transparent, actionable, and auditable, key pillars of trustworthy model in high-stakes applications.

\newpage
\bibliographystyle{plain}
\bibliography{auc_decomposition}

\end{document}